\documentclass[11pt]{article}

\usepackage{color}
\usepackage{amssymb,verbatim,ifthen}
\usepackage[all]{xy}
\usepackage{graphicx}
\usepackage{amsmath}
\usepackage{amssymb}
\usepackage[normalem]{ulem}
\usepackage{comment}

\usepackage{amsthm}
\usepackage{amssymb}
\usepackage[singlespacing]{setspace}
\usepackage{mathrsfs}
\usepackage{dsfont}
\usepackage{bbm}
\usepackage{soul}
\usepackage{ulem}
\usepackage{setspace}
\usepackage[all]{xy}
\usepackage{amsmath, amsthm, amssymb, amsfonts, enumerate}
\usepackage{ifpdf}
\usepackage{color}
\usepackage{graphicx}
\usepackage{epstopdf}
\usepackage{hyperref}
\usepackage{marginnote}
\usepackage[marginparsep=3cm]{geometry}

\usepackage{todonotes}

\usepackage{mathtools}

\newcounter{figurecounter}
\usepackage{multirow}

\usepackage[backend=bibtex,citestyle=authoryear,bibstyle=authoryear,natbib=true,labelnumber,doi=false,isbn=false,url=false,maxnames=20,
giveninits=true,maxcitenames=3,sorting=nyt,terseinits=false,clearlang=true]{biblatex} 

\DeclareFieldFormat{labelnumberwidth}{[#1]} 
\defbibenvironment{bibliography}  
{\list
	{\printtext[labelnumberwidth]{%
			\printfield{prefixnumber}%
			\printfield{labelnumber}}}
	{\setlength{\labelwidth}{\labelnumberwidth}%
		\setlength{\leftmargin}{\labelwidth}%
		\setlength{\labelsep}{\biblabelsep}%
		\addtolength{\leftmargin}{\labelsep}%
		\setlength{\itemsep}{\bibitemsep}%
		\setlength{\parsep}{\bibparsep}}%
	}
{\endlist}
{\item}

\DeclareNameAlias{default}{last-first} 

\renewbibmacro*{volume+number+eid}{%
	\printfield{volume}%
	\setunit*{\addnbthinspace}
	\printfield{number}%
	\setunit{\addcomma\space}%
	\printfield{eid}}
\DeclareFieldFormat[article]{number}{\mkbibparens{#1}}

\renewbibmacro{in:}{} 

\AtEveryBibitem{\clearfield{issue}} 
\AtEveryBibitem{\clearfield{month}} 
\AtEveryBibitem{\clearlist{language}}

\newtheorem*{theorem*}{Theorem}

\def\split{true}

\ifthenelse{\equal{\split}{true}} {

\textwidth = 450pt \oddsidemargin = 2pt \evensidemargin = 2pt

} {
\usepackage{fullpage}

}
\newtheorem{theorem}{Theorem}

\newtheorem{corollary}[theorem]{Corollary}

\newtheorem{lemma}{Lemma}

\newtheorem{definition}{Definition}
\newtheorem{example}{Example}

\newtheorem{remark}{Remark}

\def\xi{{\rm{R}}}

\newcommand{\ignore}[1]{}

\newcommand{\ep}{\varepsilon}
\newcommand{\dN}{{\mathbb{N}}}
\newcommand{\dR}{\mathbb{R}}

\begin{document}

\pagestyle{plain} \lineskip=1pt\baselineskip=15pt\lineskiplimit=2pt

\title{Choosing a consultant in a dynamic investment problem
\thanks{Lehrer acknowledges the support of grants ISF 591/21 and DFG
KA 5609/1-1.
Solan acknowledges the support of the ISF grant 211/22.}
}
\author{Yuval Cornfeld\thanks{School of Mathematical Sciences,
Tel Aviv University, Tel Aviv 69978, Israel. e-mail:
\textsf{cornfeld@mail.tau.ac.il}.}, Ehud Lehrer\thanks{Department of Economics, Durham University, Durham, DH13LB, UK,  \textsf{ehud.m.lehrer@durham.ac.uk}.} 
,  Eilon Solan\thanks{School of Mathematical Sciences,
Tel Aviv University, Tel Aviv 69978, Israel. e-mail:
\textsf{eilons@post.tau.ac.il}.}}
 \maketitle

\thispagestyle{empty}

  \bigskip\bigskip\bigskip
\noindent{\textsc{Abstract}}:
\begin{quote}
Consider a dynamic decision-making scenario where at every stage the investor has to choose between investing in one of two projects or gathering more information. 
At each stage, the investor may seek counsel from one of several consultants, 
who, for a fixed cost, provide partial information about the realized state. 
We explore the optimal strategy and its dependence on the belief and the consultation cost. Our analysis reveals that if one of the consultants discloses the state with a nonzero probability, this consultant will be 
used in any optimal strategy,
provided the consultation cost
 is sufficiently small. 
\end{quote}

\bigskip

\newpage
\lineskip=1.8pt\baselineskip=18pt\lineskiplimit=0pt \count0=1

\section{Introduction}

Decision problems, 
in which a decision maker (DM) has to select an action when the state is unknown, 
are abundant.
In many cases, to improve her performance, 
the DM can use the services of consultants who, for a certain fee, 
provide information on the unknown state of nature.
Sometimes, 
to further improve her performance,
the DM can approach a certain consultant several times,
or approach different consultants one after the other.
Since different consultants provide information of different types and qualities, and charge different fees,
the DM's optimal strategy,
which dictates which consultant to approach as a function of information gained so far and when to make a decision, may be difficult to derive.

In this paper, we study a simple decision problem with a dynamic choice of consultants,
and derive properties of the optimal strategy.

Two setups that fall naturally into our model
are investments and medical diagnosis.
Before venture capital funds decide whether to make an investment, they consult various experts about the future prospects of the company they consider investing in.
Similarly, before recommending a treatment,
doctors run various tests,
which provide statistical information regarding the patient's situation.
In both of these examples, 
the next expert to be consulted (resp., the next test to be run) may depend on the information provided by earlier experts (resp., tests).

To fix ideas, consider an investor who has two possible investment opportunities, $R$ and $L$, whose profitability depends on the state of nature:
investment $R$ (resp., $L$) yields a profit when the state of nature is $r$ (resp., $\ell$), and $0$ otherwise.
The initial prior that the state of nature is $r$ is $p_0$.
The investor has several consultants at her disposal.
For a fixed cost $c$, a consultant provides information about the state of nature, and thereby about the profitability of each of the investment opportunities. 
This information is given by an experiment \`a la Blackwell, 
whose outcome is conditionally independent of the outcome of past experiments done by the consultant or by other consultants.
At every stage,
the investor can approach one of the consultants, or, if she deems the information the consultants may provide not worth the cost, she can select one of the investment opportunities.

We study the optimal strategy of the investor and the value of the decision problem as a function of both the prior belief and the cost of consulting. We prove that if there is a consultant who, with positive probability, reveals the state of nature, then, provided that the consultation cost is below a certain threshold, in \emph{all} optimal strategies this consultant will be consulted at least once.

Consultants who may reveal the state of nature arise naturally in, e.g., military and competitive intelligence.
Suppose the DM is a country (or a firm), that is looking for information on future military activities of an enemy country (or on technological developments of a competing firm).
The country can employ various intelligence gathering method, such as imagery intelligence, cyber intelligence, signal intelligence, human intelligence, and covert operations. Each means of obtaining information can be thought of as a different consultant.
Some of these methods provide only probabilistic information, while others may sometimes provide decisive information about the enemy's plans (such as the breaking of the enigma in World War II).

We also prove that under a technical condition on the signaling probabilities of the consultants,
the value function is piecewise linear,
and, as a function of the initial belief,
there is a finite number of possible optimal strategies.
This property facilitates the task of finding the optimal strategy.

Finally, we study a restricted investment problem where the consultants have two possible types: revealers and estimators.
A consultant is a \emph{revealer} if with some probability $t > 0$ it reveals the state,
and with the remaining probability $1-t$ it provides no information.
A consultant is an \emph{estimator} if it provides a signal that with probability $q$ matches the state
and with probability $1-q$ mismatches the state.
We show that in a special symmetric case,
where the initial belief is $\frac{1}{2}$ and all consultants are revealers and estimators,
the optimal strategy is either 
(i) to consult no consultant, and immediately select an action,
or
(ii) to selects one consultant and repeatedly consult him until making a decision.

\bigskip
\paragraph{Related literature}
The idea of a sequential test goes back to \cite{dodge1929method}, who proposed the idea of a double-sampling inspection procedure. \cite{neyman1933ix}, greatly advanced the subject by providing an instrument to determine the effectiveness of the different tests. \cite{wald1992sequential}, developed sequential hypothesis testing, which is the basis for our type of decision problem, collecting information to distinguish between two possible states. Using the theory of dynamic programming, \cite{bellman1956problem}, made the calculation of optimal solutions possible. 
\cite{Herman-Chernoff-d768f788-2220-3b7c-95d1-6f942055e7f9}, \citeyear{chernoff1972sequential}, \citeyear{chernoff1973approaches}, and \cite{whittle1964some}, \citeyear{whittle1965some}, focused on asymptotically optimal solutions when the cost goes to zero, as well as results in the related area of bandit problems. 
\cite{raiffa1961applied}, applied the methods of sequential hypothesis testing to the field of statistical decision-making and information acquisition. 
\cite{hellman1969learning}, limited the DMs to strategies with finite memory. Most early papers on sequential hypothesis testing assume that the number of stages is strictly bounded (where the bound is known or unknown), which we do not.

\bigskip
A recent paper exploring sequential decision problems pertinent to our research is \cite{naghshvar2013active}. While their model whose model is more general than ours, it studies a different question, specifically, 
the bounds on information acquisition rate. Their paper shows that an upper bound can be obtained via an analysis of two heuristic strategies for a dynamic selection of actions. One strategy that achieves asymptotic optimality, where the notion of asymptotic optimality, due to Chernoff, implies that the relative difference between the total cost achieved by the proposed policy and the optimal total cost approaches zero as the penalty of wrong investment increases. The second heuristic strategy is shown to achieve asymptotic optimality only in a limited setting such as the problem of a noisy dynamic search. 
See also \cite{wilson2014bounded}, whose model focuses on a decision maker with finite memory, \cite{boehm2020theoretical}, whose model examines the optimal strategies for investment problems with dynamic reward rates in dynamic environments, and a decision criterion that changes over the course of the decision process, 
and \cite{zhang2022analytical}, whose model examines dynamic decision-making with a continuous unknown parameter or state, a methodology focusing on the continuation-value functions created by feasible continuation strategies.

\bigskip
\paragraph{The structure of the paper} In Section \ref{section:model} we introduce the model of a sequential investment problem. In Section \ref{section: Fundamental properties} we present the fundamental properties of the investment problem. Section \ref{section: Revealing consultants} provides the results for consultants that reveal the state. Section \ref{section: rational ratio} provides a sufficient condition that ensures the value function is piecewise bilinear in the prior and the cost. Section \ref{section:Extensions} provides the results for a special family of consultants.

\section{The Model and the Main Results}
\subsection{The Model}
\label{section:model}

There are two state of nature $\Omega = \{r,\ell\}$
and two actions $A = \{R,L\}$;
action $a \in A$ yields the gain $u(a,\omega)$ in state $\omega \in \Omega$. Actions are interpreted as investment opportunities.
Investment $R$ (resp., $L$) yields a profit $u(R,r)$ (resp., $u(L,\ell)$) when the state of nature is $r$ (resp., $\ell$), and $0$ otherwise. We assume w.l.o.g.~that the maximum between $u(R,r)$  and $u(L,\ell)$ is $1$. The state of nature is $r$ with probability $p_0$ (and $\ell$ with probability $1-p_0$).
There are $m$ consultants.
Each consultant $j$ is characterized by a function $S_j : \Omega \to \Delta(S)$,
where $S$ is some given finite set of signals, and $\Delta(S)$ is the set of probability distributions over $S$.

At every stage, the investor can either (a) select a consultant $j \in J := \{1,2,\ldots,m\}$, 
pay a fixed amount $c > 0$, and obtain a signal that is drawn according to $S_j(\omega)$, where $\omega$ is the state of nature, or (b) select one of the actions in $A$ and terminate the investment problem. The goal of the investor is to maximize her expected total payoff, namely, the expected gain from choosing the correct action minus the total undiscounted expected payments she made to consultants.
We assume that $c < \max\{u(R,r), u(L,\ell)\}$.
If $c \geq \max\{u(R,r), u(L,\ell)\}$, then it is optimal for the investor to never consult any consultant.

We denote the investment problems by $G=(p_0,J,c)$, where $p_0$ is the initial probability of $r$, $J$ is the set of consultants, and $c$ is the consultation cost.

A \emph{history} is a finite sequence of pairs -- a consultant and a signal. The history determines, 
through Bayes rule, the decision maker's posterior belief about the state of nature at that history.

A \emph{strategy} is a function $\sigma$ from the set of all finite histories, denoted  \emph{H}, to $A \cup J$. Denote the strategy space by $\Sigma$. A strategy is \emph{Markovian} if the choice at each history depends only on the posterior belief over $\Omega$.
Denote the expected payoff of a strategy $\sigma$ for $G=(p_0,J,c)$ by $\gamma_J(p_0,c;\sigma)$.
Note that the function $p_0 \mapsto \gamma_J(p_0,c;\sigma)$ is linear in both $p_0$ and $c$.
For an elaboration on this point, see the proof of Lemma~\ref{theorem:simple:properties}.

\subsection{Fundamental properties of the investment problem}\label{section: Fundamental properties}

In this section we present fundamental properties of the model: the existence of an optimal strategy, the linearity of the payoff of a strategy as a function of the prior and the cost, and the dynamic programming characterization of the value.
\bigskip

Given a finite set $J$ of consultants, the \emph{value} function $V_{J}: [0,1] \times (0,1)\mapsto L$ is defined by:
$$V_J(p_0,c):=\sup_{\sigma \in \Sigma}\gamma_J(p_0,c;\sigma).$$

\begin{definition}
For each $s \in S$, $\omega \in \Omega$, and $j \in J$, 
denote by $q(\omega\rvert s,j)$ the conditional probability of state $\omega$ upon receiving the signal $s$ from consultant $j$, when the prior belief is $(1/2,1/2)$: 
\[ 
q(\omega\rvert s,j)= \frac{S_j(s\rvert \omega)}{S_j(s\rvert \omega) + S_j(s\rvert \omega^c)},
\]
where
$\omega^c$ is the complementary state to the state $\omega$. \end{definition}
Note that 
$\frac{q(\omega\rvert s,j)}{q(\omega^c\rvert s,j)}= 
\frac{S_j(s\rvert \omega)}{S_j(s\rvert \omega^c)}$.
With this notation, when the prior belief is $p_0$, the posterior belief after receiving the signal $s$ from consultant $j$ is 
\begin{align}
\label{equ:posterior}
post(p_0,s,j)  &  \coloneq \frac{p_0\cdot {q(r\rvert s,j)}}{p_0\cdot {{q(r\rvert s,j)}} + (1 - p_0)\cdot {q(\ell\rvert s,j)}}\\ & =\frac{1}{1+\frac{1-p_0}{p_0}\cdot \frac{q(\ell\rvert s,j)}{q(r\rvert s,j)}}=\frac{1}{1+\frac{1-p_0}{p_0}\cdot \frac{S_j(s\rvert \ell)}{S_j(s\rvert r)}} . \nonumber
\end{align}

\begin{remark} \rm\label{posterior formula}
   (i) The posterior belief after receiving the signal $s_1$ from consultant $j$ and the signal $s_2$ from consultant $i$ is $\frac{1}{1+\frac{1-p_0}{p_0}\cdot \frac{q(\ell\rvert s_1,j)}{q(r\rvert s_1,j)}\frac{q(\ell\rvert s_2,i)}{q(r\rvert s_2,i)}}$.
In particular, the posterior belief after receiving the signal $s$ from consultant $j$ for $n$ consecutive stages is 
$\frac{1}{1+\frac{1-p_0}{p_0}\cdot (\frac{q(\ell\rvert s_1,j)}{q(r\rvert s_1,j)})^n}=
\frac{p_0\cdot {q(r\rvert s,j)}^n}{p_0\cdot {{q(r\rvert s,j)}^n} + (1 - p_0)\cdot {q(\ell\rvert s,j)}^n}$.\\ \vskip .3cm

    (ii) 
It is well known that, 
in terms of the 
  log-likelihood ratio, updating of belief is additive:
\begin{equation}
\label{equ:log}
   \ln\left(\frac{post(p_0,s,j)}{1-post(p_0,s,j)}\right)
= \ln\left(\frac{p_0}{1-p_0}\right) + \ln\left( \frac{q(r \mid s,j)}{1-q(r \mid s,j)}\right)= \ln\left(\frac{p_0}{1-p_0}\right) + \ln\left(\frac{S_j(s\rvert r)}{ S_j(s\rvert \ell)}\right).
   \end{equation} 
\end{remark}

The following result lists several simple properties of the value function. The proof is standard and appears in  Appendix~\ref{appendix:theorem:simple:properties}.

\begin{lemma}
\label{theorem:simple:properties}
For every investment problem, a Markovian optimal strategy exists.
Moreover, $V_J$ is continuous, convex in $p_0$ for every fixed $c$,
convex and monotonically decreasing in $c$ for every fixed $p_0$,
and satisfies the following recursive equation:
\begin{equation}
\label{equ:recursive}
V_J(p_0,c)=\max_{j \in J}\left\{p_0 u(R,r),(1-p_0)u(L,\ell),
\sum_{s \in S} P_j(p_0,s)\cdot V_J(post(p_0,s,j),c) -c\right\}.
\end{equation}
where $P_j(p_0,s)=p_0 S_j(s|\omega) + (1-p_0)S_j(s|\omega^c)$ is the probability that when the prior is $p_0$, the signal provided by consultant $j$ is $s$.
\end{lemma}
 
For every fixed $c$, the value function $V_J$ is continuous and convex. 
This implies that when $p_0$ is sufficiently high, $p_0 u(R,r)\ge 
\max_{j \in J}\left\{\sum_{s \in S} P_j(p_0,s)\cdot V_J(post(p_0,s,j),c) -c\right\}$. 
Consequently, there is a cutoff point  ${p_R}<1$ such that 
every Markovian optimal strategy selects $R$ when $p_0\in ({p_R},1]$. Likewise, there is a threshold $0<{p_L}$  such that 
every Markovian optimal strategy selects $L$ when $p_0$ is in $[0,{p_L})$. When $p \in ({p_L},{p_R})$, an
optimal strategy selects one of the consultants to obtain a signal from.
We are unaware of an analytic characterization of these thresholds, and the only crude bounds we have for them are ${p_L} \ge c$ and ${p_R} \le 1-c$.
Note that the posterior belief is a martingale. Thus, as soon as an informative consultant is used infinitely often,
the posteriors converge to 0 or 1.
Therefore, when at least one of the consultants in $J$ provides information, as $c$ goes to 0, $p_L$ goes to 0 and $p_R$ goes to 1.

\begin{example}\label{graph 8}
Consider an investment problem with three signals $S = \{r,\ell,\emptyset\}$, a consultation cost of $0.01$, and two consultants, whose signaling functions are as follows:
\[
\begin{array}{lll}
S_1(r|r) = S_1(\ell|\ell) = 800/1000, & S_1(\ell|r) = S_1(r|\ell) = 200/1000, & S_1(\emptyset|r) = S_1(\emptyset|\ell) = 0,\\
S_2(r|r) = S_2(\ell|\ell) = 625/1000, & S_2(\ell|r) = S_2(r|\ell) = 35/1000, & S_2(\emptyset|r) = S_2(\emptyset|\ell) = 340/1000.
\end{array}
\]
A strategy is optimal whenever its selection is as follows:
\[ 
\begin{array}{|c|c|}
\hbox{Range} & \hbox{Action}\\
\hline
p \leq 0.025 & L\\
0.025 \leq p \leq 0.088 & \hbox{consultant 1}\\
0.088 \leq p \leq 0.367 & \hbox{consultant 2}\\
0.367 \leq p \leq 0.633 & \hbox{consultant 1 or consultant 2}\\
0.633 \leq p \leq 0.912 & \hbox{consultant 2}\\
0.912 \leq p \leq 0.975 & \hbox{consultant 1}\\
0.975 \leq p & R\\
\hline
\end{array}
\]
Thus, there are a continuum of optimal strategies:
at each belief in the interval $[0.367,0.633]$,
it is optimal to consult either consultant.
The value function and optimal strategies are displayed in Figure~\thefigurecounter;
Each color in the graph represents the optimal actions for that belief: red (respectively, green, blue, black) corresponds to consulting consultant 1 
(respectively, consultant 2, either consultant, choosing an investment).

Similarly, for every $n \in \mathbb{N}$, there is an investment problem with $n$ consultants, 
such that all are used in any optimal strategy,
and for each consultant there is an optimal strategy that chooses it at $p_0=1/2$.
\end{example}
\color{black}

\includegraphics[width=\textwidth]{{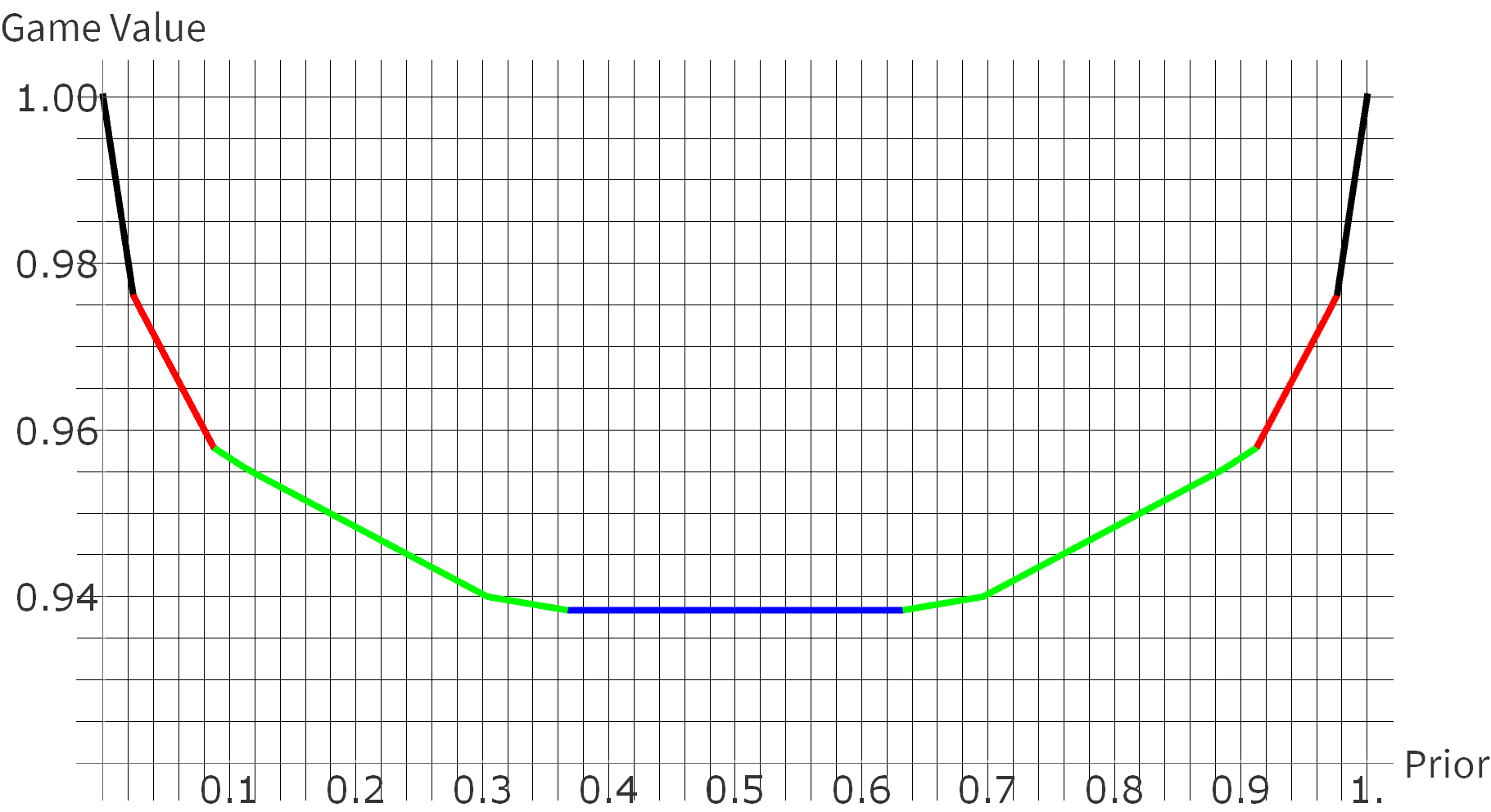}}\\
Figure \thefigurecounter: The value function related to Example \ref{graph 8} and the ranges in which different consultants are used.
\setcounter{figurecounter}{\thefigurecounter+1}

\bigskip

The following example exhibits two effects of reducing the consultation cost:
the value function increases, 
and the frequency of consultants' utilization rises.

\begin{example}\label{graph 7}
Consider an investment problem with three signals $S = \{r,\ell,\emptyset\}$ and two consultants:
\[
\begin{array}{lll}
S_1(r|r) = S_1(\ell|\ell) = 0.8, & S_1(\ell|r) = S_1(r|\ell) = 0.2, & S_1(\emptyset|r) = S_1(\emptyset|\ell) = 0,\\
S_2(r|r) = S_2(\ell|\ell) = 0.05, & S_2(\ell|r) = S_2(r|\ell) = 0, & S_2(\emptyset|r) = S_2(\emptyset|\ell) = 0.95.
\end{array}
\]
Receiving a signal from the first consultant enables the investor to update her belief, yet she remains uncertain about the true state. The second consultant, however, reveals the true state with a probability of $0.05$, while keeping the investor's belief unchanged with a probability of $0.95$.

Figure~\thefigurecounter\ exhibits the value function for various consultation costs. In this figure, each colored line represents the value function for a different cost, from $c=0.02$ to $0.3$. For $c=0.3$ (the bottom green line) the investor never consults a consultant. As the cost $c$ diminishes, the number of different optimal strategies rises, increasing the number of linear segments within the value function.
The expected payoff of each strategy is linear with respect to the prior. 
Thus, for any fixed cost, each linear segment in the graph corresponds to a different strategy, and the intersection between two segments corresponds to a change of the optimal strategy (a black dot represents a change of strategy). As the intersection
between two segments is a change of strategy, 
when the prior is in the interior of a segment
and the optimal strategy indicates to select consultant 1, the posterior is in the interior of another segment.

\includegraphics[width=0.9\textwidth]{{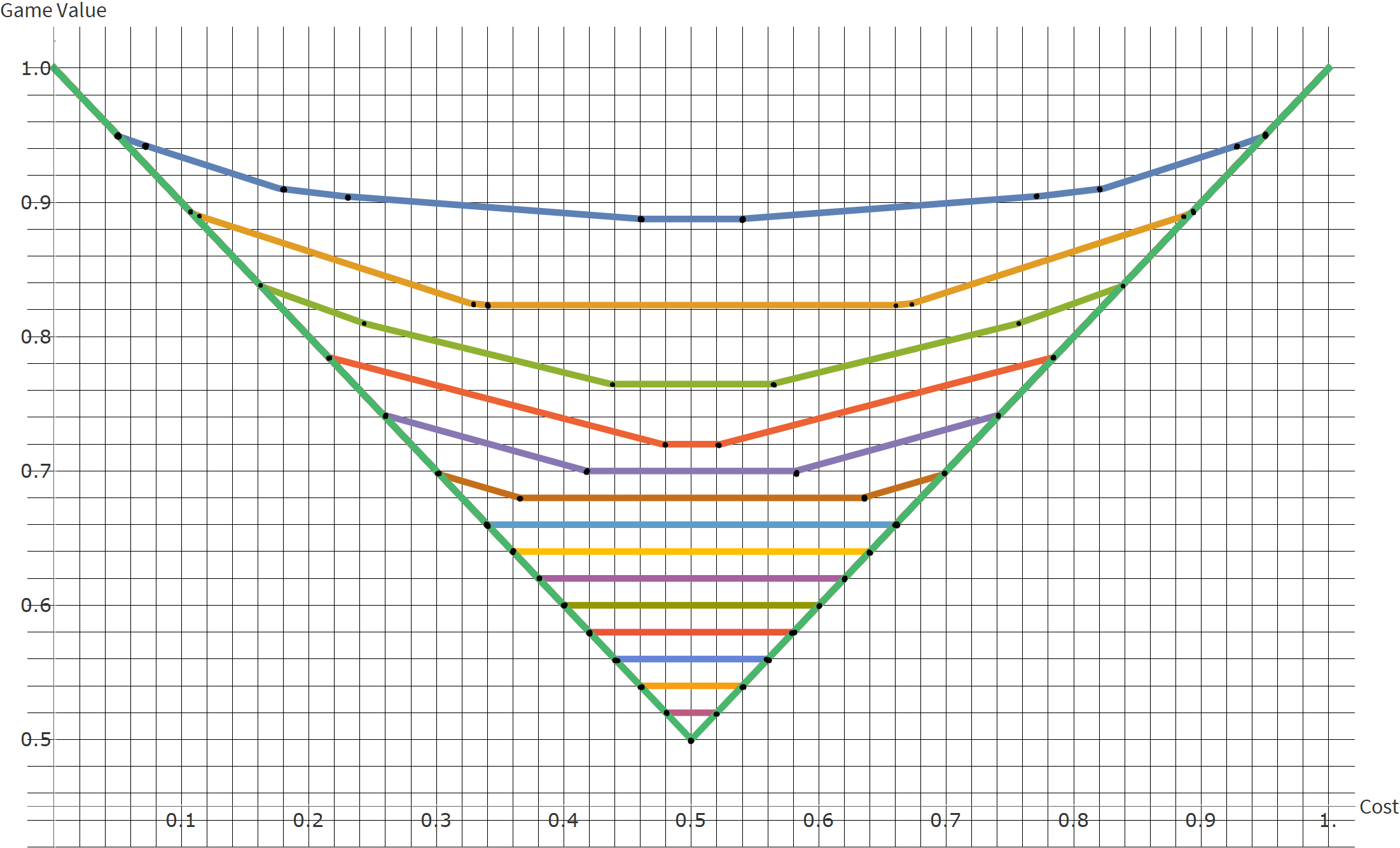}}
\includegraphics[width=0.1\textwidth]{{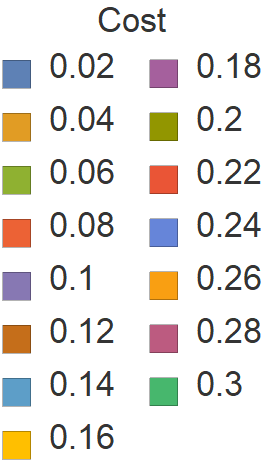}}
\centerline{Figure \thefigurecounter: The value function of the investment problem in Example \ref{graph 7}.}
\end{example}

\setcounter{figurecounter}{\thefigurecounter+1}

\subsection{Revealing consultants}\label{section: Revealing consultants}

In this section, we present the concept of revealing signals and study their role in the optimal strategy.
A consultant is called \emph{revealing} if with positive probability, the signal that he provides reveals the state of nature.

\begin{definition}
Let $j \in J$, $s \in S$, and $\omega \in \Omega$.
The signal $s$ is \emph{$\omega$-revealing} by consultant $j$ if $S_j(s \mid \omega) > 0$ and $q(\omega \mid s,j) = 1$. 
Such a signal is called \emph{revealing by $j$}. 
A consultant $j$ who has an $\omega$-revealing signal, for every $\omega \in \Omega$, is called \emph{revealing}.
\end{definition}

Note that the signal that reveals the state of nature may depend on the state.
Examples of revealing signals are positive biopsies, intelligence reports from well-placed agents, and striking oil when searching for oil reservoirs.

Will the DM use revealing consultants in her optimal strategy?
The answer seems to depend on the probability by which the revealing signals are provided.
If these signals are provided with high (resp., low) probability, 
the revealing consultant will (resp., will not) be used.
As the next result states, the probability by which the revealing signals are provided should be compared to the consultation cost:
if this cost is low, the revealing consultant will be used.

\begin{theorem}\label{revealer,estimator theorem}
For every prior $p_0 \in [0,1]$, every revealing consultant $j_*$, and every set of consultants $J^-$ that are not revealing, there exists $C>0$ such that for every $c\le C$ all optimal strategies in the investment problem $G=(p_0,J^- \cup \{j_*\},c)$ consult $j_*$ at least once.
\end{theorem}

\begin{proof}
Let $\ep>0$ be a lower bound on the probability that the consultant $j_*$ sends a revealing signal. 
Let $\sigma_*$ be the strategy that consults $j_*$ until it reveals the state of nature,
and then selects the action that matches the state.
Since 
the distribution of the revelation stage is dominated by a geometric distribution with parameter $\ep$,
\[ \gamma_{\{j_*\}}(p_0,\sigma_*,c) \geq p_0 u(R,r) + (1-p_0) u(L,\ell) - \frac{c}{\ep}. \]

We now provide an upper bound on the expected payoff given by consulting only consultants in $J^-$. 
Denote $q=\max_{s \in S,j \in J^-,b \in B}q(b\rvert s,j) < 1$. 
Let $\sigma_n$ be a strategy that $n$ times consults a consultant from $J^-$, and then selects the more favorable action.
The posterior belief (the probability that the state is $r$) after stage $n$ lies
between $\frac{p_0\cdot {q}^n}{p_0\cdot {q}^n + (1 - p_0)\cdot {(1 - q)}^n}$ and $\frac{p_0\cdot {(1-q)}^n}{p_0\cdot {(1-q)}^n + (1 - p_0)\cdot {q}^n}$. 
Therefore, for every strategy $\sigma$ that consults only consultants in $J^-$,
$\gamma_J(p_0,c;\sigma)$ cannot be higher than 
\begin{align*}	
\max_{n \in \mathbb{N}}\Biggl[p_0
u(R,r)\left(\frac{p_0\cdot q^n}{p_0\cdot {q}^n + (1 - p_0)\cdot {(1 - q)}^n}\right) - & c \cdot n, \\
(1-p_0)u(L,\ell) & \left(1-\frac{p_0\cdot {(1-q)}^n}{p_0\cdot {(1-q)}^n + (1 - p_0)\cdot {q}^n}\right)- c \cdot n 
\Biggl].
\end{align*}

Provided $c$ is sufficiently small, this quantity is at most
$p_0 u(R,r) + (1-p_0) u(L,\ell) - \frac{c}{\ep}$,
which implies that strategies that never consult $j_*$ are worse than $\sigma_*$. 
Hence, all optimal strategies must consult $j_*$. 
\end{proof}

\subsection{Consultants with a rational ratio}\label{section: rational ratio}

According to Lemma~\ref{theorem:simple:properties},
the value function $V_J$ is convex. As Example~\ref{example:strictly convex} below shows, 
this function may be strictly convex on $[p_L,p_R]$.
In this section we show that under some conditions,
$V_J$ is piecewise linear.

\begin{definition}
\label{def:rational}
A set of consultants $J$ has a \emph{rational ratio} if there exists a real number $Q>0$ such that for each $s \in S$ and $j \in J$,
the ratio 
$\ln\left(\frac{S_j(s|\ell)}{S_j(s|r)}\right)=\ln\left(\frac{q(r\rvert s,j)}{q(\ell\rvert s,j)}\right)$ 
is an integer multiple of $Q$.
\end{definition}

\begin{remark}\label{remark rational}
Recall Eq.\ (\ref{equ:posterior}) and Remark \ref{posterior formula}.
    The conditions $$\ln(q(\ell\rvert s_1,j)/q(r\rvert s_1,j))=m \cdot Q,$$ 
    $$\ln(q(\ell\rvert s_2,i)/q(r\rvert s_2,i))=n \cdot Q,$$ 
    where $m$ and $n$ are both positive integers (or both negative integers), 
imply that 
$$\frac{p_0\cdot {q(r\rvert s_1,j)}^n}{p_0\cdot {{q(r\rvert s_1,j)}^n} + (1 - p_0)\cdot {q(\ell\rvert s_1,j)}^n}=\frac{1}{1 + \frac{1 - p_0}{p_0}\cdot e^{nmQ}}=\frac{p_0\cdot {q(r\rvert s_2,i)}^m}{p_0\cdot {{q(r\rvert s_2,i)}^m} + (1 - p_0)\cdot {q(\ell\rvert s_2,i)}^m}.$$ That is, obtaining $n$ times the signal $s_1$ from consultant $j$ yields the same posterior belief as obtaining $m$ times the signal $s_2$ from consultant $i$.
    Similarly, the conditions $$\ln(q(\ell\rvert s_1,j)/q(r\rvert s_1,j))=m \cdot Q,$$ $$\ln(q(\ell\rvert s_2,i)/q(r\rvert s_2,i))=-n \cdot Q,$$ where $m$ and $n$ are positive integers, imply that obtaining $n$ times the signal $s_1$ from consultant $j$ and then obtaining $m$ times the signal $s_2$ from consultant $i$ yields the posterior belief $p_0$.
\end{remark}

\begin{theorem}\label{finite posterior piecewise}
Let $G=(p_0,J,c)$ be an investment problem, and let $J$ be a set of consultants with a \emph{rational ratio}. 
Then 
(i) there is a finite set $\mathcal{P} \subset [p_L,p_R]$ such that for every history,
 if the posterior belief at history $h$ is in $[p_L,p_R]$, then it is in  $\mathcal{P}$, and 
(ii) the value function is piecewise bilinear in $p_0$.

\end{theorem}

\begin{proof} Recall 
Eq.~\eqref{equ:log}, and 
let $Q$ be the real number in Definition~\ref{def:rational}.
The difference $\ln\left(\frac{q_t}{1-q_t}\right)-\ln\left(\frac{p_0}{1-p_0}\right)$,
where $q_t$ is the posterior belief at stage $t$, is an integer multiple of $Q$.
Since for every belief $p>1-p_R$ or $p<p_L$, the optimal strategy at that belief is to choose an investment,
part (i)  follows.
To establish (ii), let us revisit the optimal strategy's general structure. Recall that at any belief $p < p_L$ (or $p > p_R$), the optimal strategy entails selecting $L$ (or $R$). Given this and from part (i),
 any history (consisting of past chosen consultants and the random signal they generated) that did not follow taking $L$ or $R$ as actions, corresponds to one of finitely many posteriors.
Moreover, each such history translates into a discrete movement along the log-likelihood scale, incrementing or decrementing by integer multiples of $Q$. This movement mirrors the adjustments made to beliefs, either towards $\ell$ or $r$. Thus, 
a pure Markovian strategy can be equivalently defined on the distance, measured in terms of $Q$, between the actual posterior and the prior, traced along the log-likelihood scale. 
Since this distance is uniformly bounded (for all the priors in the range of $[p_L,p_R]$), only a finite number of options exist for {such} a pure Markovian strategy.
{The payoff corresponding to each of these strategies is linear in $p_0$.}
 Hence, the value function is the maximum among a finite set of linear functions, resulting in a piecewise linear one. \end{proof}

The following example shows that when the consultants do not have a rational ratio, the value function may not be piecewise linear.

\begin{example}
\label{example:strictly convex}
Consider 
an investment problem with two signals $S = \{a,b\}$ and  one consultant,
whose signaling function is given by:
\[
S_1(a|r) = x, \ \ \ S_1(b|r) = 1-x, \ \ \ S_1(a|\ell) = y, \ \ \ S_1(b|\ell) = 1-y,
\]
where $0 < y < x < 1$.
Denote $C_1 := \ln(x/y)$ and $C_2 :=\ln((1-x)/(1-y))$,
and
assume that $C_1/C_2$ is an irrational number,
so that the consultant does not have a rational ratio. 
We show that for $c$ sufficiently small, the value function is strictly convex on $[p_L,p_R]$.

Since $C_1/C_2$ is irrational, the set $W:= \{nC_1 - mC_2 \colon n,m \in \dN\}$ is dense in $\dR$.
Assume $c$ is small enough so that when obtaining twice the signal $a$ (resp., $b$) when the prior is $p_L$
(resp., $p_R$),
the posterior is still in $(p_L,p_R)$.
This implies that for every 
prior $p_0 \in (p_L,p_R)$ and every posterior $q \in (p_L,p_R)$ 
satisfying $\ln\left(\frac{q}{1-q}\right) - \ln\left(\frac{p_0}{1-p_0}\right) \in W$,
there is a history of signals,  $h$, such that the sequence of posteriors along that history remains in $(p_L,p_R)$, and the posterior after $h$ is $q$.

To show that the value function is strictly convex, fix two distinct beliefs $p_0,q \in (p_L,p_R)$,
and let $\sigma$ be a pure Markovian optimal strategy at the prior $p_0$.
We will show that $\sigma$ is not optimal at the prior $q$.
Indeed, by the discussion above,
there is a history of signals $h$ such that 
(i) when the prior is $p_0$, the posterior beliefs along $h$ are all in $(p_L,p_R)$,
and
(ii) when the prior is $q$, the posterior beliefs after $h$ is not in $[p_L,p_R]$,
while the posterior belief after any strict prefix of $h$ is in $(p_L,p_R)$.
The strategy $\sigma$ is not optimal given the prior $q$, because it advises seeking further information after $h$. However, since it leads to a posterior beyond the range of $[p_L,p_R]$, the optimal response would be to choose one of the investments instead.
\end{example}

\subsection{Three-signal investment problem}
\label{section:Extensions}

In this section, we consider a limited set of consultants, which can provide three signals: $S = \{r, \ell, \emptyset\}$.
The signals $r$ and $\ell$ are positively correlated with the state of nature, while the signal $\emptyset$ provides no information on the state. 
One example for such consultants is medical tests, such as the Covid self-test kits, which provide three signals -- positive, negative, or inconclusive.
We further assume that the problem is symmetric:
the probability of obtaining the signal $r$ when the state is $r$ is the same as the probability of obtaining the signal $\ell$ when the state is $\ell$, and the payoffs for investing in \hbox{R} (resp., \hbox{L}) when the state is $r$ (resp., $\ell$) is 1 and 0 otherwise.

We will see that when the set of available consultants consists only of such consultants, one can derive stronger structural properties of the optimal strategy.

For each consultant $j$, denote by $1-t_j$ the probability that $j$ provides the signal $\emptyset$ in either state, by $q_j \cdot t_j$ the probability that $j$ provides the signal that matches the state, and by $(1-q_j) \cdot t_j$ the probability that $j$ provides the signal that does not match the state, 
see Figure \thefigurecounter(A).
We will identify a consultant $j$ with the pair $(q_j,t_j)$.

We assume w.l.o.g.~that the signal is positively correlated with the state,\footnote{Otherwise, the investor can invert the meaning of the signal.} that is, $q_j > 0.5$. 
The assumption that the probability of the signal that matches (resp., does not match) the state is independent of the state means that the consultant has no bias among the states.

Two extreme types of consultants are the \emph{estimator} who is never silent, that is, $t_j=1$, and provides a probabilistic estimation of the state; 
and the \emph{revealer} who gives a revealing signal or a noninformative signal,
that is, $q_j=1$.
The signaling functions of an estimator 1 and a revealer 2 are, then,
\[
\begin{array}{lll}
S_1(r|r) = S_1(\ell|\ell) = q_1, & S_1(\ell|r) = S_1(r|\ell) = 1-q_1, & S_1(\emptyset|r) = S_1(\emptyset|\ell) = 0,\\
S_2(r|r) = S_2(\ell|\ell) = t_2, & S_2(\ell|r) = S_2(r|\ell) = 0, & S_2(\emptyset|r) = S_2(\emptyset|\ell) = 1-t_2.
\end{array}
\]
In Example~\ref{graph 8}, consultant 1 is an estimator, and
in Example~\ref{graph 7}, consultant 1 is an estimator and consultant 2 is a revealer.

\begin{remark}
The parameter $t_j$ delays the rate at which consultant $j$ provides information. 
Since (i) payoffs are not discounted, 
(ii) the noninformative signal $\emptyset$ does not change the belief on $\Omega$,
and (iii) there is a Markovian optimal strategy,
for the purpose of calculating the value
and the optimal strategy, a $(q_j,t_j)$-consultant with cost $c$ is equivalent to a $(q_j,1)$-consultant with cost $c^*:=\frac{c}{t_j}$.
In particular, to analyze three-signal investment problems,
it is w.l.o.g.~to assume that all consultants are estimators or revealers, albeit with a different consultation cost. Similarly, for an investment problem with consultants with varying consultation costs, an equivalent investment problem with a common consultation costs can be created by adjusting the probability of the noninformative signal.
\end{remark}

If there were more than one estimator or more than one revealer, an optimal strategy would use only one of each group: the estimator $j$ with the highest $q_{j}$ and the revealer $i$ with the highest $t_{i}$.

The next lemma states that the set of beliefs at which it is optimal to use a revealer is convex.
When consulting a revealer,
the belief changes only when the revealer reveals the state.
Note that if it is optimal to consult a reveal $j$ at a certain belief $p$,
then $V_J(p,c) = 1-\frac{c}{t_j}$.

\begin{lemma}
Let $J$ be a set of consultants in a three-signal investment problem that includes a revealer $j$ with parameter $t_j$, and let $c \in (0,1)$. 
The set of beliefs $p$ where $V_J(p,c) = 1-\frac{c}{t_j}$ is convex.
\end{lemma}

\begin{proof}
Suppose that $V_{J}(p_*,c)=1-\frac{c}{t_j}$.
The symmetry of the problem implies that $V_{J}(1-p_*,c)=1-\frac{c}{t_j}$.
The convexity of the value function implies that $V_{J}(p,c)\leq 1-\frac{c}{t_j}$ for every $p \in (p_*,1-p_*)$,
while since the strategy that always consults $j$ is available to the DM, $V_{J}(p,c)\geq 1-\frac{c}{t_j}$ for every $p \in (p_*,1-p_*)$.
The claim follows.
\end{proof}

As consultants in three-signal investment problems are symmetric, Lemma \ref{theorem:simple:properties} implies the following properties of the value function.

\begin{corollary}
For every fixed cost, as a function of the prior, the value function is convex, symmetric around 1/2, monotone non-increasing from 0 to 1/2, and monotone non-decreasing from 1/2 to 1.
\end{corollary}

\begin{remark}
    If there were more than one estimator or more than one revealer, an optimal strategy would use only one of each group, that is, the estimator $j$ with the highest $q_{j}$ and the revealer $i$ with the highest $t_{i}$.
\end{remark}

As we now show, the monotonicity of the value as a function of the prior
implies that in the presence of a revealer, 
when the prior belief is close to $\frac{1}{2}$,
the optimal strategy consults only the revealer.

\begin{lemma}\label{lemma revealer payoff}
Let $J$ be a set of consultants in a three-signal investment problem that includes a revealer $j$ with parameter $t_j$, and let $c \in (0,1)$. If there is a belief $p_*$ such that $V_{J}(p_*,c)=1-\frac{c}{t_j}$, then for every $p \in [p_*,1-p_*]$ there is an optimal strategy that only consults the revealer, and $V_{J}(p,c)=1-\frac{c}{t_j}$.
\end{lemma}

\begin{proof}

The strategy that consults consultant $j$ until it reveals the state yields the payoff $1-\frac{c}{t_j}$.
Since $V_{J}(p_*,c)=1-\frac{c}{t_j}$, the symmetry of the problem implies that $V_{J}(1-p_*,c)=1-\frac{c}{t_j}$. The convexity of the value function implies that $V_{J}(p_*,c)=1-\frac{c}{t_j}$ for every $p \in [p_*,1-p_*]$, and the second claim follows.
The first claim holds since $1-\frac{c}{t_j}$ is the payoff that corresponds to the strategy that always consults the revealer with parameter $t_j$.
\end{proof} 

When consulting a revealer,
the belief changes only when the revealer reveals the state.
Since there is an optimal Markovian strategy,
we obtain that when the initial belief is $1/2$,
either there is an optimal Markovian strategy that always consults the revealer, or there is an optimal Markovian strategy that never consults him.

Lemma~\ref{lemma revealer payoff} allows us to strengthen Theorem~\ref{revealer,estimator theorem} for three-signal investment problems when the prior is $1/2$ as follows.

\begin{corollary}
Let $G=(1/2,J,c)$ be a three-signal investment problem, and suppose that one of the consultants is a revealer. One of the following statements holds: 

\begin{itemize}
\item There is an optimal strategy that does not consult any consultant.
\item There is an optimal strategy that only consults the revealer.
\item There is an optimal strategy that never consults the revealer.
 \end{itemize}
\end{corollary}

\setcounter{figurecounter}{\thefigurecounter+1}

\begin{corollary}\label{revealer superiority}
    For each prior $p_0 \in [0,1]$, each $q$, and each $t$, there is a cost $C$ such that for every $c\le C$ the optimal strategy in the investment problem $G=(p_0,\{(q,1),(1,t)\},c)$ is to only consult the revealer.
\end{corollary}

In three-signal investment problems, even though the set of signals is greatly limited, there is still no clear ranking between different consultants. 
If a consultant has a higher $q$ and a higher probability of providing information $t$ than another consultant, then the investor will never choose the latter. But if one consultant has higher $q$ and the other provides information more often, the identity of the better consultant depends on the prior and the cost, as well as the other consultants in the investment problem. Therefore, the optimal strategy may use both consultants. 
As Example~\ref{graph 8} shows, there can be consultants that can both be used under an optimal strategy in a range of priors.
The following example shows that a revealer and an estimator can provide the same value in a range of priors.
In particular, it exhibits two different investment problems having the same value function, and thus the value function is not uniquely associated with one investment problem.

\begin{example}\label{rational ratio example}
Let $p_0$ be 1/2. Consider the consultants $j_1=(q_1=0.8,t_1=1)$ and $j_2=(q_2=16/17,t_2=17/50)$. 
Define $\rho_t$ (resp., $\lambda_t$) to be the number of times until stage $t$ in which the signal is $r$ (resp., $\ell$). Let $G_1=(1/2,\{j_1\},c)$ and $c$ be such that the optimal strategy in $G_1$  is to consult $j_1$ until 
$|\rho_t - \lambda_t| = 2$. 
The expected number of stages until
$|\rho_t - \lambda_t| = 2$
is $2/(0.8^2+0.2^2) = 50/17$. 
When this happens, 
the posterior belief is $16/17$ or $1/17$,  
and therefore the value is $16/17-c \cdot 50/17$. 

Consider now the investment problem $G_2=(1/2,\{j_2\},0.05)$. The optimal strategy is to consult $j_2$ until one gets a non-\hbox{Silent} signal once. 
When this occurs, 
the posterior belief is either $16/17$ or $1/17$, and the expected number of stages to get a (non-\hbox{Silent}) signal is $50/17$.
Therefore, the problem $G_2$ has the same value as $G_1$. 

In fact, these two investment problems have the same value for a neighborhood of $p_0=1/2$ and $c=0.05$.
\end{example}

\let\cleardoublepage\clearpage

\printbibliography[heading=bibintoc]

\appendix

\section{Proofs}

\subsection{Proof of Lemma~\ref{theorem:simple:properties}}
\label{appendix:theorem:simple:properties}

An optimal strategy exists since payoffs are bounded by 1
and since the payoff is a continuous function of the strategy in the product topology.

To prove that $V_J$ is convex in $p_0$ for every fixed $c$,
and convex in $c$ for every fixed $p_0$, 
we start by showing that $\gamma_\sigma(p_0,J,c)$ is bilinear in $p_0$ and $c$.
Indeed, 
denote by $P_r$ (resp., $P_\ell$) the probability that under $\sigma$ the investor chooses $R$ (resp., $L$) conditional that the state is $r$ (resp., $\ell$). 
Denote by $E_r$ (resp., $E_\ell$) the expected number of stages until termination conditional that the state is $r$ (resp., $\ell$). 
Note that $P_r$, $P_\ell$, $E_r$, and $E_\ell$ are independent of $p_0$ and $c$.
With these notations,
\begin{equation}
\label{equ:payoff}
\gamma_J(p_0,c;\sigma)=u(R,r) \cdot P_r \cdot  p_0+ u(L,\ell) \cdot P_\ell \cdot  (1-p_0)-\bigl(E_r\cdot p_0 +E_\ell\cdot (1-p_0)\bigr) \cdot  c,
\end{equation}
which is bilinear in $p_0$ and $c$.

The convexity properties of $V_J$ now follow since this function is the maximum of functions that are bilinear in $p_0$ and $c$.

Eq.~\eqref{equ:payoff} implies that $\gamma_\sigma$ is non-increasing in $c$ for every fixed $\sigma$, $J$, and $p_0$,
and hence $V_J(p_0,\cdot)$,
as the maximum of functions that are non-increasing in $c$,
is non-increasing in $c$.

Eq.~\eqref{equ:recursive} follows from Bellman's equation.

We finally show that $V_J$ is continuous.
Indeed, on $(0,1) \times (0,1)$ continuity of $V_J$ follows from its convexity.
For $p_0=1$, continuity follows since $V_J(1,c) = u(R,r)$  and since by Eq.~\eqref{equ:recursive} $v_J(p,c) = pu(R,r)$ whenever $p$ is sufficiently close to 1.
Continuity of $V_J$ at $p_0=0$ follows analogously.

\end{document}